\documentclass[10pt, journal, twocolumn]{IEEEtran}

\usepackage[utf8]{inputenc}		
\usepackage[scr=boondoxo,frak=euler,bb=ams]{mathalfa}
\usepackage{amsmath,amsthm,amsfonts,amssymb} 	
\usepackage{graphicx,graphics}
\usepackage{float}
\usepackage{tikz}
\usetikzlibrary{shapes, snakes, patterns}
\usepackage{xcolor}		

\newcommand{\spheading}[2][5em]{
	\rotatebox{90}{\parbox{#1}{\raggedright #2}}}
\usepackage[export]{adjustbox}
\usepackage[justification=centering]{caption}
\usepackage{subcaption}
\DeclareCaptionLabelSeparator{periodspace}{.\quad}
\usepackage{epstopdf}
\usepackage{enumerate,enumitem} 
\usepackage{cite}
\usepackage{suffix}
\usepackage{mathtools}
\usepackage{tabularx}
\usepackage{booktabs}		
\usepackage{boldline,multirow,diagbox,hhline} 
\usepackage{relsize}
\usepackage{cuted}
\usepackage{titlesec}
\usepackage{chngcntr}
\usepackage{refcount}

\usepackage{textcomp}
\usepackage{algorithmic}
\def\BibTeX{{\rm B\kern-.05em{\sc i\kern-.025em b}\kern-.08em
		T\kern-.1667em\lower.7ex\hbox{E}\kern-.125emX}}

\theoremstyle{definition}

\newtheorem{theorem}{Theorem}[]
\newtheorem{proposition}[theorem]{Proposition}

\newtheorem{definition}[]{Definition}

\newtheorem{remark}[]{Remark}

\newcommand{\PreserveBackslash}[1]{\let\temp=\\#1\let\\=\temp}
\newcolumntype{C}[1]{>{\PreserveBackslash\centering}p{#1}} 
\newcolumntype{R}[1]{>{\PreserveBackslash\raggedleft}p{#1}} 
\newcolumntype{L}[1]{>{\PreserveBackslash\raggedright}p{#1}} 

\newcolumntype{Y}{>{\centering\arraybackslash}b{0.7cm}}
\newcolumntype{M}{>{\centering\arraybackslash}b{2.1cm}}
\newcolumntype{Z}{>{\centering\arraybackslash}b{1.75cm}}
\newcolumntype{G}{>{\centering\arraybackslash}m{1.5cm}}
\newcolumntype{Q}{>{\centering\arraybackslash}b{3.75cm}}
\newcolumntype{A}{>{\centering\arraybackslash}m{1.5cm}}
\newcolumntype{B}{>{\centering\arraybackslash}m{0.6667cm}}
\newcolumntype{T}{>{\centering\arraybackslash}b{6.5ex}}
\newcolumntype{S}{>{\centering\arraybackslash}b{13.25ex}}
\newcolumntype{?}{!{\vrule width 1pt}}
\newcolumntype{+}{?@{\hskip 1pt}?}

\DeclareFontFamily{U}{mathx}{}
\DeclareFontShape{U}{mathx}{m}{n}{<-> mathx10}{}
\DeclareSymbolFont{mathx}{U}{mathx}{m}{n}
\DeclareMathAccent{\widecheck}{0}{mathx}{"71}

\definecolor{bluee}{RGB}{0, 82, 200}
\definecolor{redd}{RGB}{245, 30, 30}
\definecolor{greenn}{RGB}{0, 150, 62}

\makeatletter
\newcommand\footnoteref[1]{\protected@xdef\@thefnmark{\ref{#1}}\@footnotemark}
\makeatother

\newcommand{\mbf}[1]{\boldsymbol{\mathrm{#1}}}

\newcommand{\MleqK}{${}^{\star}_{}$}

\newcommand{\commhide}[1]{}

\newcommand{\psp}{\hspace{0.1em}}
\newcommand{\pspp}{\hspace{0.05em}}
\newcommand{\nsp}{\hspace{-0.1em}}
\newcommand{\nspp}{\hspace{-0.05em}}

\makeatletter
\newcommand{\vast}{\bBigg@{4}} 
\newcommand{\Vast}{\bBigg@{5}} 
\newcommand{\csize}[1]{\bBigg@{#1}} 
\makeatother

\makeatletter
\newcommand{\ostar}{\mathbin{\mathpalette\make@circled\star}}
\newcommand{\make@circled}[2]{\ooalign{$\m@th#1\smallbigcirc{#1}$\cr\hidewidth$\m@th#1#2$\hidewidth\cr}}
\newcommand{\smallbigcirc}[1]{ \vcenter{\hbox{\scalebox{0.77778}{$\m@th#1\bigcirc$}}}}
\makeatother

\newcommand{\summ}{\textstyle\sum\limits}

\newcommand{\neighb}{{\rotatebox[origin=c]{90}{\footnotesize{$\rangle\nsp\langle\!$}}}}
\newcommand{\neighbdsets}{\mathcal{D} \pspp\, \neighb \,\pspp\widecheck{\mathcal{D}}}
\newcounter{tempeqncnt}

\title{On the Optimal Number of Grids for Differentially Private Non-Interactive $K$-Means Clustering}

\author{
	\IEEEauthorblockN{
		Gokularam Muthukrishnan\textsuperscript{*}, Anshoo Tandon\textsuperscript{*}
		\thanks{\noindent\textsuperscript{*}The authors are with the Center of Data for Public Good, Foundation for Science Innovation and Development, Indian Institute of Science, Bengaluru 560012, India (e-mail: gokularam.m@datakaveri.org; anshoo.tandon@gmail.com).
        }}
}

\begin{document}
	
	\maketitle	
    
	\begin{abstract}
        Differentially private $K$-means clustering enables releasing 
        cluster centers derived from a dataset while protecting the privacy of the individuals. 
        Non-interactive 
        clustering
        techniques based on privatized histograms are attractive because the released data synopsis can be reused for other downstream tasks without additional privacy loss. The choice of the number of grids for 
        discretizing 
        the data points is crucial, as it directly controls the 
        quantization
        bias and the amount of noise injected to preserve privacy. The widely adopted strategy selects a grid size that is independent of the number of clusters and also relies on empirical tuning. 
        In this work, we revisit this choice and propose a refined grid-size selection rule 
        derived by minimizing an upper bound on the expected deviation in the $K$-means objective function, 
        leading to a more principled discretization strategy for non-interactive private clustering. 
        Compared to prior work, our
        grid resolution differs both in its dependence on the number of clusters and in the scaling with dataset size and privacy budget.
        Extensive numerical results
        elucidate
        that the proposed strategy results in accurate clustering compared to the state-of-the-art 
        techniques, 
        even under tight privacy budgets.
	\end{abstract}
	
	\begin{IEEEkeywords}
        Differential privacy, 
        $K$-means clustering, 
        Non-interactive data analysis,
        Grid-size selection,
        Laplace mechanism.
	\end{IEEEkeywords}

	\section{Introduction}\label{sec:intro}

    \IEEEPARstart{D}{ata}
    clustering \cite{kevin2022probabilistic} is a fundamental unsupervised machine learning technique that, when applied to datasets 
    with
    attributes from several users, identifies meaningful user groups 
    and reveals
    the latent structures within the data.
    The \textit{$K$-means clustering} 
    is a
    widely adopted 
    clustering method for both low and high-dimensional data (e.g., \cite{von2007tutorial}). However, the user
    attributes often contain sensitive information,
    %
    making it is critical to ensure that no individual's data can be distinguished or re-identified from the clustering.
    Differential privacy (DP) \cite{dwork2014algorithmic} 
    is a strong and widely adopted 
    framework for preserving user privacy \cite{fioretto2025differential}, and over the years, several differentially private clustering algorithms have been developed.
    %
    %
    %
    %

    %
    %
    Ghazi \textit{et al.} \cite{ghazi2020differentially} provided the approximation guarantee for DP clustering.
    Private $K$-means clustering has been studied under the distance-based privacy notion in \cite{epasto2023k} and under the estimation error-based notion in \cite{nguyen2016privacy}. 
    Authors of \cite{imola2023differentially} studied private
    hierarchical clustering of graphs%
    , where privacy of the vertex interactions is preserved.
    A genetic algorithm-based DP $K$-means clustering scheme has been proposed in \cite{zhang2013privgene} 
    and
    a private variant of the $K$-means++ algorithm has been proposed in \cite{nock2016k}. 
    The sample-and-aggregate framework has been used for private clustering in \cite{nissim2007smooth} and \cite{wang2015differentially}, though their approach relied on strong assumptions on the input data.
    In \cite{jones2021differentially} and \cite{nguyen2021differentially}, DP $K$-means algorithms based on maximum coverage were proposed%
    , showing guarantees comparable to \cite{balcan2017differentially}%
    , but the results  quite far from those of the non-private Lloyd’s algorithm%
    . 
    Recently, scalable implementations of private $K$-means clustering for very large datasets have been studied in \cite{cohen2022scalable} and \cite{cohen2022near}.
    In \cite{gilad2025differentially}, cluster-wise description of the dataset is obtained in a differentially private manner.

    Broadly speaking, 
    private clustering 
    algorithms 
    can be classified into two categories, namely, \textit{interactive} and \textit{non-interactive}. 
    Interactive algorithms 
    deplete
    the entire privacy budget for the clustering task, and the data cannot be reused for any other purpose without further degradation of privacy guarantee.
    In contrast, non-interactive algorithms (also known as synopsis-based approaches) first generate a private summary of the data distribution, which can be reused for multiple queries or downstream tasks without additional privacy leakage.

    \textit{DPLloyd}%
    \cite{blum2005practical,mcsherry2009privacy}%
    , the private version of Lloyd algorithm for $K$-means clustering%
    ,
    is the prominent example of interactive schemes. 
    It requires 
    careful distribution 
    of the available privacy budget across iterations: too few iterations prevent convergence, while too many increase the amount of noise injected to preserve privacy. 
    IBM's \texttt{diffprivlib} \cite{holohan2019diffprivlib} implements DPLloyd with a budget allocation strategy proposed in \cite{su2017differentially}%
    , and decreasing exponential budget distribution is used in \cite{dwork2011firm}.
    Another private version of Lloyd algorithm has been proposed in \cite{lu2021differentially}, where the centroids are resampled in each iteration through exponential mechanism and the convergence is guaranteed.

    Chang \textit{et al.} \cite{chang2021locally} proposed a 
    non-interactive
    %
    %
    private clustering 
    scheme
    %
    %
    based on private \textit{coreset} construction
    \cite{feldman2009private,cohen2021new}
    %
    for the local model of DP%
    %
    , where the dataset is 
    privately compressed into a smaller, weighted set of points 
    through recursive locality-sensitive hashing 
    \cite{andoni2017lsh}
    %
    and a classical, non-private clustering algorithm is executed on them%
    ; the framework is later extended to the central model in \cite{chang2021practical}%
    .
    %
    This method outperforms \cite{balcan2017differentially} and scales well for massive datasets, 
    but requires several hyperparameters to be set \textit{a priori}, necessitating extensive parameter tuning
    to achieve best results%
    .
    %

    %
    Uniform grid 
    non-interactive 
    methods partition the data domain into equal-sized grids, providing a simple and efficient structure for releasing private 
    histograms
    \cite{qardaji2013differentially,su2017differentially}.
    %
    %
    One of the most prominent works in this category is the \textit{Extended Uniform Grid $K$-Means} (EUGkM) algorithm \cite{su2017differentially}, which has been shown to outperform interactive DP clustering algorithms 
    in various settings,
    particularly in 
    lower dimensions.
    %
    %
    However, its construct is not directly motivated by the clustering objective, rather 
    adapted 
    from \cite{qardaji2013differentially}, which focused on range queries%
    ; consequently, the selected grid size  
    is independent of the number of clusters%
    , and this often leads to degenerate scenarios where 
    the grid resolution is too coarse to even distinguish individual clusters%
    , especially in the high privacy regime%
    . 
    %
    %
    Furthermore, the 
    grid size
    depends on a hyperparameter obtained through extensive 
    empirical 
    tuning, and may fail to maintain consistent performance across 
    different 
    configurations.
    %

    %
    %
    In this article, we derive the optimal 
    grid size
    for uniform grid non-interactive clustering that effectively balances the 
    discretization bias
    and the amount of noise injected 
    for
    privacy.
    \begin{itemize}
        \item We 
        determine 
        the optimal
        grid 
        size
        that minimizes an upper bound on the expected deviation in the $K$-means objective function. While some existing works 
        \cite{zhang2012functional,chaudhuri2011differentially} 
        consider objective perturbation to guarantee privacy, we \textit{control} the perturbation in the objective function.

        \item Our principled approach 
        effectively 
        addresses 
        EUGkM's limitations %
        and show clear benefits%
        : 
        Our 
        grid size 
        is 
        directly
        dictated by the number of clusters, and it also differs from that of EUGkM in the exponents governing the dataset size and the privacy budget.

        \item Extensive numerical results demonstrate that our solution 
        achieves superior clustering accuracy compared to 
        state-of-the-art techniques, including leading open-source DP clustering implementations.
    \end{itemize}

    \section{Background and Motivation}\label{sec:background}
    

    %
    %
    %
    %
    
    We briefly recapitulate 
    several
    key concepts%
    , in part to establish notations and conventions that will be followed in this article, and motivate our proposed grid selection rule.


 

    \subsection{Differential Privacy}

    Let $\mathcal{D}\in\mathcal{X}^{N}_{}$ denote the dataset that contains data records 
    $\mbf{x}_{l}^{}\in\mathcal{X}$, ${l}\in\mathbb{N}_{N}^{}$,
    collected from $N$ individuals. 
    We term a pair of datasets $(\mathcal{D},\widecheck{\mathcal{D}})$ that differ by only a single data record as neighboring
    (or adjacent) 
    datasets, and we 
    denote
    $\neighbdsets$; 
    %
    %
    %
    %
    %
    %
    in particular, we consider add-remove model 
    %
    \cite{abadi2016deep,zhu2019poission}%
    , where $\widecheck{\mathcal{D}}$ can be constructed by adding or removing 
    a single
    data record from $\mathcal{D}$, i.e., $\big|\mathcal{D} \setminus \widecheck{\mathcal{D}}\big| 
    + \big|\widecheck{\mathcal{D}} \setminus \mathcal{D}\big| = 1$.
    %
    %

    %
    %
    A trusted central curator holds the sensitive dataset $\mathcal{D}$ and is responsible for
    securely releasing the 
    result of the query ${f}:\mathcal{X}^{N}_{}\to\mathcal{Y}$ on dataset $\mathcal{D}$.
    DP guarantee ensures that it is difficult to 
    infer
    any individual's presence in $\mathcal{D}$ from the 
    response 
    by essentially \textit{randomizing} it; 
    the algorithm $\mathcal{M}$ that provides randomized output to the query 
    referred to as the \textit{private mechanism}.

    \begin{definition}[Differential Privacy {\cite{dwork2006calibrating}}]
        A randomized mechanism $\mathcal{M}:\mathcal{X}^{N}_{}\to\mathcal{Y}$ is said to guarantee $\epsilon$-differential privacy%
        \footnote{
        While this work focuses on $\epsilon$- DP (or \textit{pure} DP), a relaxed notion, namely $(\epsilon,\delta)$-DP (or \textit{approximate} DP) \cite{dwork2006our}, allows the indistinguishability bound to hold up to an additive slack $\delta\in[0,1]$.
        }
        \textup{(}$\epsilon$-DP in short\textup{)} if 
        for every 
		measurable set $\mathcal{E}\subseteq\mathcal{Y}$ and every 
		pair of 
		neighboring datasets $\neighbdsets$, 
		\begin{equation} \label{eq:DP}
			\mathbb{P}\{\mathcal{M}(\mathcal{D})  \in \mathcal{E}\} 
			\leq 
			e^{\epsilon}_{}\psp \mathbb{P}\{\mathcal{M}(\widecheck{\mathcal{D}})  \in \mathcal{E}\} 
            ,
		\end{equation}
		where 
		$\epsilon 
        \in\mathbb{R}_+^{}$
		is the \textit{privacy budget}. 
    \end{definition}
	%
    %
    %
    %
    The definition imposes an information-theoretic bound on an adversary's \textit{power} to discern whether the input dataset is $\mathcal{D}$ or $\widecheck{\mathcal{D}}$. 
    %
    %
    %
    The parameter $\epsilon$ quantifies the amount of information that may be revealed about any individual’s presence; smaller values of $\epsilon$ correspond to stronger privacy guarantees.

    When the query 
    output
    is a numeric vector, the straightforward way to randomize it is by adding noise. 
    In the \textit{additive noise mechanism}, DP is achieved by perturbing the query result ${f}(\mathcal{D})\in\mathbb{R}^{M}_{}$ as
    $\mathcal{M}( \mathcal{D} ) = {f}(\mathcal{D}) + \mathbf{Z}$, 
    where $\mathbf{Z}
    \in \mathbb{R}^{M}_{}$ is the noise that is sampled from a known distribution.
    Numerous additive noise mechanisms have been proposed 
    and studied
    in the literature%
    \cite{dwork2006calibrating,
    dwork2006our,
    geng2015staircase,
    geng2020tighttrunclapl,
    liu2018generalized,
    canonne2020discrete,
    sadeghi2022offset,
    muthukrishnan2023grafting,
    alghamdi2022cactus
    }%
    .
    The \textit{Laplace mechanism} \cite{dwork2006calibrating} achieves 
    DP
    by adding noise sampled from a Laplace distribution, with appropriately calibrated scale, to each coordinate of the response.
    Because of the exponential tails of the noise distribution, the Laplace mechanism
    can 
    ensure
    $\epsilon$-DP 
    with \textit{bounded privacy loss}
    \cite{tian2018selective}.

    \subsection{Private Histogram Query}

    %
    %
    %
    %
    %
    %
    %
    %
    Consider the setting where the data domain $\mathcal{X}$ is partitioned into $M$ disjoint bins.     
    The \textit{histogram query} \cite[E.g., 3.2]{dwork2014algorithmic} requests the number of data entries in the dataset that fall within each of the bins: the $j$-th coordinate of 
    the response ${f}(\mathcal{D)}\in\mathbb{Z}_{+}^{M}$ represents the number of data entries in $\mathcal{D}$ of \textit{type} $j \in \mathbb{N}_{M}^{}$.
    %
    %
    %
    %
    %
    %
    Note that the addition or deletion of a single data record can affect the count in exactly one bin, i.e., for neighbouring datasets $\neighbdsets$, ${f}(\widecheck{\mathcal{D}})={f}(\mathcal{D})\, \pm \, \mbf{e}_{M,\psp j}^{}$ for some $j \in \mathbb{N}_{M}^{}$.

    The direct release of histogram counts may risk disclosing sensitive information, particularly about individuals 
    belonging to 
    rare types.
    %
    %
    %
    %
    %
    %
    %
    %
    %
    The privatization of histogram queries is primarily achieved through the Laplace Mechanism.%
    \footnote{\label{fn:discrete_Laplace}%
    %
    %
    While the discrete Laplace 
    mechanism is the optimal $\epsilon$-DP mechanism for histogram queries \cite{ghosh2009universally,gupte2010universally,geng2016optimalstaircase}, which also ensures
    that perturbed outputs are integer-valued 
    and provides significant utility gains over Laplace 
    in the high-privacy regime (i.e., as $\epsilon \to 0$), 
    we consider the Laplace mechanism in this article to maintain consistency with \cite{su2017differentially}: Extending our analysis to the discrete 
    counterpart
    will be investigated as a part of 
    future work.
    }
    Specifically, 
    to guarantee $\epsilon$-DP, independent noise sampled from $\mathcal{L}\big(0,\frac{1}{\epsilon}\big)$ is added 
    to each bin count.

    \begin{remark}
    \label{rem:error_counts}
    Due to inherent
    structural disjointness
    of the histogram query,
    the scale of 
    Laplace
    noise 
    required for $\epsilon$-DP
    is independent of the number of bins 
    as well as the size of the dataset. 
    Consequently, 
    the Laplace mechanism
    offers fairly high accuracy:
    The probability that any bin count 
    deviates 
    larger than $\frac{1}{\epsilon}\log\big(\frac{M}{\tau}\big)$ is at most $\tau$ 
    %
    \cite[Thm. 3.8]{dwork2014algorithmic}, 
    %
    i.e., 
    the magnitude of the error scales only logarithmically with the number of bins.
    %
    %
    %
    %
    \end{remark}

    Histograms provide a granular summary of the data distribution, and, once privatized, 
    they can be reused arbitrarily many times without further privacy leakage.
    Hence, privatized histograms facilitate 
    \textit{non-interactive analyses}, where the released private data synopses are reused for subsequent queries on the dataset and other downstream tasks (
    e.g., 
    differentially private selection 
    and synthetic data generation
    \cite{dwork2014algorithmic,fioretto2025differential}%
    ).
    In this work, we 
    focus on 
    clustering data points based on their
    noisy histograms, specifically addressing the 
    optimal grid resolution needed 
    to balance discretization (or quantization) error and the error in the bin counts 
    due to 
    injected noise (see Remark \ref{rem:error_counts}).


    \subsection{Private $K$-Means Clustering}
    %
    %
    Data Clustering is an unsupervised learning task of categorizing the points in the given dataset $\mathcal{D}=\{ \mbf{x}_{l}^{}\in
    \mathbb{R}^d_{}
    , \, {l}
    \in \mathbb{N}_{N}^{}
    \}$ into $K\leq N$ groups called clusters, $\left\{ \mathcal{C}_i^{} \subset \mathbb{N}_{N}^{} \right\}_{i=1}^K$,
    such that points within the same cluster are more `similar' to each other than to points in different clusters.
    The \textit{$K$-means clustering} \cite{kevin2022probabilistic} achieves such a partition by minimizing the average within-cluster sum of squares (WCSS), i.e., 
    \begin{equation}\label{eq:obj_kmeans}
        \min_{\mathcal{C}_1^{},\, \ldots,\, \mathcal{C}_{K}^{}}^{}
        \,
        \frac{1}{K}
        \sum_{i=1}^{K} J_i^{}
        ,
    \end{equation}
    where 
    ${J}_i^{} = \sum_{\mbf{x}_{l}^{} \in \mathcal{C}_i^{}} \Vert \mbf{x}_{l}^{} - \mbf{\mu}_{i}^{}\Vert^2_{}$ 
    and 
    $\mbf{\mu}_{i}^{} = \frac{1}{|\mathcal{C}_i^{}|} \sum_{\mbf{x}_{l}^{} \in \mathcal{C}_i^{}}^{} \mbf{x}_{l}^{}$
    are respectively the WCSS and centroid corresponding to the $i$-th cluster $\mathcal{C}_{i}^{}$, $i\in\mathbb{N}_{K}^{}$.
    %
    %
    
    %
    %
    %
    \textit{Lloyd's algorithm} is the most common technique for solving the $K$-means problem. When initialized with a set of $K$ centroids, it iteratively assigns each data point to its nearest centroid and then updates 
    each cluster’s centroid as the mean of the points assigned to it.
    For differentially private $K$-means clustering, we consider bounded data, i.e., $\mbf{x}_{l}^{}\in [-r,r]^d_{} \,\ \forall \, {l} \in \mathbb{N}_{N}^{}$, and the goal is to release the $K$ cluster centroids such that 
    any single data point has a bounded influence on them.
    
    
    \subsubsection{DPLloyd}
    %
    %
    %
    \textit{DPLloyd}
    \cite{blum2005practical,mcsherry2009privacy}
    is the popular interactive $K$-means clustering approach based on the standard Lloyd's algorithm. During each iteration, 
    %
    noise is injected 
    into the counts and the 
    sums of data points assigned to each cluster before updating the centroids. The main drawback of DPLloyd is that the privacy budget $\epsilon$ must be split across a fixed number of iterations; consequently, the noise added in each step increases with the number of iterations, leading to performance degradation without achieving convergence \cite{su2017differentially}.
    
    
    
    \subsubsection{Non-Interactive DP $K$-Means Clustering}  
    %
    %
    In this work, we 
    consider
    non-interactive
    approach to private $K$-means clustering. The data domain $[-r,r]^d_{}$ is partitioned into $M$ 
    equal-sized grids, and the histogram comprising the number of data points falling into each grid is obtained. Once the grid counts are privatized through Laplace mechanism,\footnoteref{fn:discrete_Laplace} weighted Lloyd's algorithm is executed on the grid centres using the corresponding noisy counts as weights.
    The choice of $M$ is critical%
    : 
    A larger $M$ 
    %
    provides a finer resolution, thereby reducing
    the \textit{bias} in the centroid estimates
    but increases the \textit{variance} as the dimension of the noise vector increases.

    The \textit{Extended Uniform Grid $K$-Means} (EUGkM) 
    algorithm 
    \cite{su2017differentially} 
    suggests selecting
    the number of grids as
    \begin{equation}\label{eq:M_EUGkM}
        M=\left(\frac{N\epsilon}{10}\right)^{\frac{2d}{2+d}}_{}
        ,
    \end{equation}
    generalizing the uniform gridding framework 
    developed for two-dimensional geospatial data
    \cite{qardaji2013differentially}%
    .
    While this offers a structured approach for private histogram release, the derivation 
    %
    was primarily focused on
    optimizing 
    %
    range query's accuracy
    and relies 
    on 
    hyperparameter tuning rather than 
    %
    being directly informed by 
    the cluster structure underlying the data.
    Consequently, this choice of $M$ remains independent of the number of clusters $K$, which can lead to suboptimal 
    or even degenerate 
    %
    results 
    in many
    scenarios%
    , such as when the grid resolution is too coarse to distinguish $K$ distinct centroids (i.e., $M < K$)%
    .
    

    We propose a refined grid-size selection rule that explicitly accounts for the number of clusters $K$. Specifically, we determine the optimal number of grids by minimizing 
    %
    an upper bound on the 
    expected deviation in the $K$-means objective function 
    %
    arising from
    both discretization error and the noise introduced to ensure privacy. 
    In particular, we show that the optimal grid 
    resolution 
    scales as 
    $
    M = \Theta\Big( (N \epsilon)^{\frac{2d}{4+d}} \, K^{\frac{4}{4+d}} \Big)
    $%
    ,
    %
    differing from the EUGkM rule
    %
    %
    in both its explicit dependence
    on
    $K$
    %
    and
    the exponents governing 
    $N$
    and
    $\epsilon$.
    We term this grid selection strategy as \textit{RUGNIK} (Refined Uniform Gridding for Non-Interactive $K$-means clustering).
    %
    %
    The derivation follows.


    \section{RUGNIK Derivation} \label{sec:derivation}

    Consider a cluster $\mathcal{C}_i^{}$, $i\in \mathbb{N}_K^{}$, with centroid 
    $\mbf{\mu}_i^{}
    $ and the corresponding WCSS ${J}_i^{}$. 
    The data domain $[-r,r]^d_{}$ is partitioned uniformly into $M = m^d_{}$ grids. 
    %
    %
    %
    %
    Let $\mathcal{G}_i^{}$ denote the set of all grids associated%
    \footnote{We say the grid $g$ with centre $\mbf{s}\in\mathbb{R}^d_{}$ is associated to cluster $\mathcal{C}_i^{}$ if $ \Vert \mbf{s} - \mbf{\mu}_i^{} \Vert < \Vert \mbf{s} - \mbf{\mu}_{k}^{} \Vert \,\ \forall\, {k}\in\mathbb{N}_{K}^{} \setminus \{i\}$.}
    to cluster $\mathcal{C}_i^{}$, and let $g_{ij}^{}\in \mathcal{G}_i^{}$, $j \in \mathbb{N}_{|\mathcal{G}_i^{}|}^{}$, denote the $j$-th grid, with centre $\mbf{s}_{ij}^{}\in\mathbb{R}^d_{}$, that is associated
    to cluster $\mathcal{C}_i^{}$.
    Also, let $t_{ij}^{}$ denote the number of data points falling within the grid $g_{ij}^{}$, $t_{ij}^{} = |\{\mbf{x}_q^{}\in g_{ij}^{}\}|$, and let $n_{ij}^{}$ be the corresponding noisy grid counts,
    $n_{ij}^{} = t_{ij}^{} + Z_{ij}^{}$, where $Z_{ij}^{} \sim \mathcal{L}\big(0,\frac{1}{\epsilon}\big)$. 
    Hence, the \textit{observed} WCSS for cluster $\mathcal{C}_i^{}$ is
    \begin{align*}
        \widetilde{J}_i^{} 
        =  \sum_{g_{ij}^{} \in \mathcal{G}_i^{}}^{} n_{ij}^{} \, \Vert \mbf{s}_{ij}^{} - \mbf{\mu}_i^{}\Vert^2_{}
        =  \sum_{g_{ij}^{} \in \mathcal{G}_i^{}}^{} (t_{ij}^{} + Z_{ij}^{}) \, \Vert \mbf{s}_{ij}^{} - \mbf{\mu}_i^{}\Vert^2_{}
        \\
        \implies
        \widetilde{J}_i^{} 
        =  \sum_{g_{ij}^{} \in \mathcal{G}_i^{}}^{} \sum_{\mbf{x}_q^{} \in g_{ij}^{}}^{} \Vert \mbf{s}_{ij}^{} - \mbf{\mu}_i^{}\Vert^2_{}
        + \sum_{g_{ij}^{} \in \mathcal{G}_i^{}} Z_{ij}^{} \, \Vert \mbf{s}_{ij}^{} - \mbf{\mu}_i^{}\Vert^2_{}
        .
    \end{align*}
    Expanding $J_i^{}$ as 
    $
    \sum_{g_{ij}^{} \in \mathcal{G}_i^{}} \sum_{\mbf{x}_q^{} \in g_{ij}^{}}^{} \Vert\mbf{x}_q^{} - \mbf{\mu}_i^{}\Vert^2_{}$, the deviation in the WCSS of the cluster $\mathcal{C}_i^{}$
    can be expressed as
    %
    \begin{align*}
        \widetilde{J}_i^{} - J_i^{}
        &
        = 
        \sum_{g_{ij}^{} \in \mathcal{G}_i^{}}^{} Z_{ij}^{}\psp \Vert \mbf{s}_{ij}^{} - \mbf{\mu}_i^{}\Vert^2_{}
        \\ & \qquad \quad \ \ 
        - \sum_{g_{ij}^{} \in \mathcal{G}_i^{}}^{} \sum_{\mbf{x}_q^{} \in g_{ij}^{}}^{} \big( \Vert\mbf{x}_q^{} - \mbf{\mu}_i^{}\Vert^2_{} -\Vert\mbf{s}_{ij}^{} - \mbf{\mu}_i^{}\Vert^2_{} \big) 
        ,
    \end{align*}
    where the 
    second term corresponds to the deviation due to quantization.

    Let $\mbf{b}_{q,ij}^{} \triangleq \mbf{x}_q^{} - \mbf{s}_{ij}$
    and
    $\mbf{Y}_{ij}^{} \triangleq \mbf{s}_{ij}^{} - \mbf{\mu}_i^{}$. 
    Thus, $
        \Vert\mbf{x}_q^{}\Vert^2_{} 
        = \Vert\mbf{s}_{ij}^{} + \mbf{b}_{q,ij}^{}\Vert^2_{} 
        = \Vert\mbf{s}_{ij}^{}\Vert^2_{} + 2\, (\mbf{b}_{q,ij}^{})^{\top}_{} \mbf{s}_{ij}^{} + \Vert\mbf{b}_{q,ij}^{}\Vert^2_{}
        \implies
        \Vert\mbf{x}_q^{} - \mbf{\mu}_i^{}\Vert^2 - \Vert\mbf{s}_{ij}^{} - \mbf{\mu}_i^{}\Vert^2_{}
        = \Vert\mbf{x}_q^{}\Vert^2_{} - \Vert\mbf{s}_{ij}^{}\Vert^2_{} - 2\,(\mbf{x}_q^{} - \mbf{s}_{ij}^{})^{\top}_{} \mbf{\mu}_i^{}
        = 2\, (\mbf{b}_{q,ij}^{})^{\top}_{} \mbf{Y}_{ij}^{} +  \Vert\mbf{b}_{q,ij}^{}\Vert^2_{}
    $.
    The space $[-r,r]^d_{}$ contains $K$ clusters. Assuming equally sized and uniformly spread clusters, each cluster is (roughly) contained in a hypercube of side $\frac{2r}{K^{1/d}_{}}$.
    Hence, 
    
    
    \begin{equation}\label{eq:Y_ij}
    \Vert\mbf{Y}_{ij}^{}\Vert_{\infty}^{} \lesssim \frac{r}{K^{1/d}_{}} \implies \Vert\mbf{Y}_{ij}^{}\Vert_2^{} \lesssim \frac{r}{K^{1/d}_{}} \sqrt{d}
    .    
    \end{equation}

    \begin{table*}[!t]
        \normalsize
        %
        \setcounter{tempeqncnt}{\value{equation}}
        \setcounter{equation}{\getrefnumber{eq:logical-start}-1}
        \begin{equation}\label{eq:T1T2T3}
            T_1^{} \triangleq 
            \mathbb{E}\Bigg[\Bigg|
            \sum_{i=1}^K
            \nsp
            \sum_{g_{ij}^{}\in\mathcal{G}_i^{}}^{}
            \nsp 
            Z_{ij}^{}\psp \Vert\mbf{Y}_{ij}^{}\Vert^2_{}\Bigg|\Bigg]
            , \ \
            T_2^{} \triangleq 
            \mathbb{E}\Bigg[\Bigg|
            \sum_{i=1}^K
            \nsp
            \sum_{g_{ij}^{}\in\mathcal{G}_i^{}}^{}
            \nspp
            \sum_{\mbf{x}_q^{} \in g_{ij}^{}}^{}
            \nsp
            (\mbf{b}_{q,ij}^{})^{\top}_{} \mbf{Y}_{ij}^{}\Bigg|\Bigg]
            , \ \ \text{and} \ \
            T_3^{} \triangleq 
            \sum_{i=1}^K
            \nsp
            \sum_{g_{ij}^{}\in\mathcal{G}_i^{}}^{} 
            \nspp
            \sum_{\mbf{x}_q^{} \in g_{ij}^{}}^{}
            \nsp  \mathbb{E}\big[\Vert\mbf{b}_{q,ij}^{}\Vert^2_{}\big]
        \end{equation}
        %
        \setcounter{equation}{\getrefnumber{eq:logical-next}-1}
        \begin{equation}\label{eq:T3_simplification}
            \mathbb{E}\Bigg[\Bigg|
            \sum_{i=1}^K
            \nsp
            \sum_{g_{ij}^{}\in\mathcal{G}_i^{}}^{}
            \nspp
            \sum_{\mbf{x}_q^{} \in g_{ij}^{}}^{}
            \nsp
            (\mbf{b}_{q,ij}^{})^{\top}_{} \mbf{Y}_{ij}^{}\Bigg|\Bigg]
            \leq
            \sqrt{\nspp\mathbb{E}\!\left[\!\Bigg(\nsp
            \sum_{i=1}^K
            \nsp
            \sum_{g_{ij}^{}\in\mathcal{G}_i^{}}^{}
            \nspp
            \sum_{\mbf{x}_q^{} \in g_{ij}^{}}^{} 
            \nsp
            (\mbf{b}_{q,ij}^{})^{\top}_{} \mbf{Y}_{ij}^{}
            \nsp \Bigg)^{\!\nspp2\pspp}_{}\right]} 
            \approx 
            \sqrt{\nspp\mathbb{E}\!\left[\!\Bigg(\nsp 
            \sum_{i=1}^K
            \nsp
            \sum_{g_{ij}^{}\in\mathcal{G}_i^{}}^{}
            \nspp
            \sum_{\mbf{x}_q^{} \in g_{ij}^{}}^{} 
            \nsp 
            (\mbf{b}_{q,ij}^{})^{\top}_{} \mbf{Y}_{ij}^{(0)}
            \nsp \Bigg)^{\!\nspp2\pspp}_{}\right]} 
        \end{equation}
        %
        \setcounter{equation}{\value{tempeqncnt}}
        \hrulefill
    \end{table*} 
    
    Now, the deviation in the objective function of the $K$-means clustering, 
    $\nu =
    \frac{1}{K}
    \big| \sum_{i=1}^{K}  \widetilde{J}_i^{}- \sum_{i=1}^{K} J_i^{} \big|$,
    can be expressed as
    \begin{align*}
    \nu
    &
    \nspp
    = 
    \nspp
    \frac{1}{K}\nspp
    \left|
    \sum_{i=1}^K
    \nsp
    \sum_{g_{ij}^{}\in\mathcal{G}_i^{}}^{}
    \nsp
    Z_{ij}^{}\psp\Vert\mbf{Y}_{ij}\Vert^2_{}
    - 2
    \sum_{i=1}^K 
    \nsp
    \sum_{g_{ij}^{}\in\mathcal{G}_i^{}}^{}
    \nspp
    \sum_{\mbf{x}_q^{} \in g_{ij}^{}}^{} 
    \nsp
    (\mbf{b}_{q,ij}^{})^{\top}_{} \mbf{Y}_{ij}^{}
    \right.
    \\ & \qquad \qquad \qquad \qquad \qquad \qquad \ \
    \left.
    - \sum_{i=1}^K
    \nsp
    \sum_{g_{ij}^{}\in\mathcal{G}_i^{}}^{}
    \nspp
    \sum_{\mbf{x}_q^{} \in g_{ij}^{}}^{}
    \nsp
    \Vert\mbf{b}_{q,ij}^{}\Vert^2_{}
    \right|
    .
    \end{align*}
    Therefore, 
    \begin{equation*}\label{eq:expected_nu}
    \mathbb{E}[\nu] 
    \leq  
    \frac{1}{K}\big(
    T_1^{} + 2\,T_2^{} + T_3^{}
    \big)
    ,
    \end{equation*}
    \refstepcounter{equation}%
    \label{eq:logical-start}%
    where 
    $T_1^{}$, $T_2^{}$, and $T_3^{}$ are given in \eqref{eq:T1T2T3} and
    the expectation is with respect to noise and data distribution.
    
    We now bound each term in the above summation in terms of $m$, $N$, $\epsilon$, $K$, and $d$.
    Firstly, observing that
    \begin{equation*}
    \Bigg|\sum_{i=1}^K
    \nsp
    \sum_{g_{ij}^{}\in\mathcal{G}_i^{}}^{}
    \nsp
    Z_{ij}^{}\psp\Vert\mbf{Y}_{ij}^{}\Vert^{2}_{} \Bigg| 
    \lesssim \frac{dr^{2}}{K^{2/d}} \Bigg|
    \sum_{i=1}^K
    \nsp
    \sum_{g_{ij}^{}\in\mathcal{G}_i^{}}^{}
    \nsp
    Z_{ij}^{}\Bigg|
    \end{equation*}
    and noting that
    \begin{align*}
    \mathbb{E}\Bigg[\Bigg|
    \sum_{i=1}^K
    \nsp
    \sum_{g_{ij}^{}\in\mathcal{G}_i^{}}^{}
    \nsp
    Z_{ij}^{}\Bigg|\Bigg]
    & 
    \leq
    \sqrt{\nspp\mathbb{E}\!\left[\!\Bigg(\nsp
    \sum_{i=1}^K
    \nsp
    \sum_{g_{ij}^{}\in\mathcal{G}_i^{}}^{}
    \nsp
    Z_{ij}^{} \nsp \Bigg)^{\!\nspp2\pspp}_{}\right]} 
    =\frac{\sqrt{2M}}{\epsilon}
    \end{align*}
    as a consequence of Jensen’s inequality, we have
    \begin{equation}\label{eq:T1}
        T_1^{} \lesssim \frac{dr^{2}}{K^{2/d}}\frac{\sqrt{2M}}{\epsilon}
        .
    \end{equation}
    Before proceeding further with the simplification of the remaining terms, we make a few assumptions on the data distribution.
    We assume that the data points $\{ \mbf{x}_{l}^{}\}_{l=1}^{N}$  are independent, and within each grid, the points falling in that grid are roughly uniformly distributed. Furthermore, the $i$-th cluster centroid $\mu_{i}\approx\mu_{i}^{(0)}$, where $\mu_{i}^{(0)}$ is the ground truth centroid. Hence, $[\mbf{b}_{q,ij}]_{l}\sim \mathcal{U}\big(-\frac{r}{m},\frac{r}{m}\big)$, and $\mbf{Y}_{ij}^{} \approx \mbf{Y}_{ij}^{(0)} \triangleq \mbf{s}_{ij}^{} - \mbf{\mu}_i^{(0)}$.%
    %
    \refstepcounter{equation}
    \label{eq:logical-next}
    
    
    
    Under these considerations, we deduce \eqref{eq:T3_simplification}%
    , and hence, we have
    $
        T_2^{}
        \lesssim
        \nsp\left(
        \mathbb{E}\nsp\left[\left(
            \summ_{i=1}^K \summ_{g_{ij}^{}\in\mathcal{G}_i^{}}^{} \summ_{\mbf{x}_q^{} \in g_{ij}^{}}^{}  
            (\mbf{b}_{q,ij}^{})^{\top}_{} \mbf{Y}_{ij}^{(0)}
            \right)^{\!2\psp}_{}\right]
        \right)^{\!1/2}_{}
        =
        \!\bigg(
        \summ_{i=1}^K \summ_{g_{ij}^{}\in\mathcal{G}_i^{}}^{} \summ_{\mbf{x}_q^{} \in g_{ij}^{}}^{} 
            \mathbb{E}\nsp\left[\Big( 
            (\mbf{b}_{q,ij}^{})^{\top}_{} \mbf{Y}_{ij}^{(0)}
            \Big)^{\nsp2\psp}_{}\right]
        \!\bigg)^{\!1/2}_{}
    $%
    . Furthermore, 
    $
    \mathbb{E}\nsp\left[\Big(
    (\mbf{b}_{q,ij}^{})^{\top}_{} \mbf{Y}_{ij}^{(0)}
    \Big)^{\nsp2\psp}_{}\right]
    \leq
    \mathbb{E}\nsp\left[\big\Vert\mbf{b}_{q,ij}^{}\big\Vert^{2}_{} \big\Vert\mbf{Y}_{ij}^{(0)}\big\Vert^{2}_{}\right]
    \lesssim
    \dfrac{dr^2_{}}{K^{2/d}_{}}\dfrac{dr^2_{}}{3m^2_{}}
    $%
    . Thus,
    we have
    \begin{equation}\label{eq:T3}
        T_2^{}
        \lesssim
        \sqrt{\frac{N}{3}}\frac{dr^2_{}}{m K^{1/d}_{}}    
        .
    \end{equation}
    %
    %
    %
    Additionally, 
    we have  
    \begin{equation}\label{eq:T2}
        \mathbb{E}\big[\Vert\mbf{b}_{q,ij}^{}\Vert^2_{}\big] = \frac{dr^2_{}}{3m^2_{}}
        \implies
        T_3^{}=N\frac{dr^2_{}}{3m^2_{}}
        .
    \end{equation}
    
    Therefore, from \eqref{eq:T1}, \eqref{eq:T2}, and \eqref{eq:T3},
    we have
    $\mathbb{E}[\nu] 
    \lesssim
    \frac{dr^2_{}}{K}\pspp
    h(m)
    $, where
    $
        h(m)=
        \frac{\sqrt{2} m^{d/2}_{}}{\epsilon K^{2/d}_{}}
        +\frac{2}{m K^{1/d}_{}}\sqrt{\frac{N}{3}}
        +\frac{N}{3 m^2_{}}
    $.
    Since $h(m)$ is a convex function%
    \footnote{%
    Since $
    h''(m)=
    \Big(\frac{d}{2}-1\Big) \frac{dm^{d/2-2}_{}}{\sqrt{2}\epsilon K^{2/d}_{}}
    + \frac{4}{m^3_{}K^{1/d}_{}}\sqrt{\frac{N}{3}} 
    + \frac{2N}{m^4_{}}
    \geq 0 
    \,\ \forall\, 
    d\geq 2$.
    }
    in $m$, the solution to $h'(m)=0$ minimizes the upper bound on the expected deviation in the objective function of the $K$-means clustering.
    Moreover,
    $
    h'(m)=\frac{d}{\sqrt{2}\epsilon K^{2/d}_{}}\frac{1}{m^3_{}} \xi(m)
    $, where
    \begin{equation*}
        \xi(m)=
        %
        m^{d/2+2}_{} - \rho K^{1/d}_{} m - \rho {\sqrt{\frac{N}{3}}} K^{2/d}_{} 
        , \quad
        \rho = \frac{\epsilon}{d}\sqrt{\frac{8N}{3}}
        ,
    \end{equation*}
    and hence, the optimal $m$ solves $\xi(m)=0$.
    We cannot solve this in closed form. However, from the following proposition, we know that $\xi(m)$ has a unique root.

    \begin{proposition}\label{prop:xi_root}
        $\xi$ has a unique root on $\mathbb{R}_{++}^{}$.
    \end{proposition}
    \begin{proof}
        %
        %
        %
        %
        %
        %
        %
        %
        We observe that $\xi$ is strictly convex on $\mathbb{R}_{+}^{}$; hence, $\xi'_{}$ is a strictly increasing function on $\mathbb{R}_{+}^{}$.
        Together with the facts that $\xi(0)<0$ and $\lim\limits_{m\to\infty}^{}\xi(m)=\infty$, this implies that there exists a unique ${\tau} \geq 0$ such that $\xi({\tau}) = \xi(0) < 0$ and $\xi'_{}(m)>0 \,\ \forall \, m>{\tau}$ (take $\tau=0$ if $\xi'_{}(0)\geq 0$).
        Thus, $\xi$ is strictly increasing on $({\tau},\infty)$, and since $\xi(m) < 0 \,\ \forall \, m\in[0,{\tau}]$%
        , the result follows. 
    \end{proof}    

    
    The following result presents an upper and a lower bound on the optimal number of grids under mild technical conditions on $\epsilon$.

    \begin{theorem}\label{thm:M_bound}
        The optimal number of grids is bounded as follows.
        \begin{enumerate}[label=(\roman*)]
            \item If $\epsilon \geq \epsilon_0^{} \triangleq d\sqrt{\frac{3K}{8N}}\Big(1+\sqrt{\frac{N}{3}}\Big)^{\nspp-1}_{}$,  
            $M \geq 
            \left\lfloor
            \big(
                {\eta}
                K^{2/d}_{}
                \big)^{\frac{2d}{4+d}}_{}
            \right\rfloor
            $%
            , and
            \item if $\epsilon \leq 
            \epsilon_1^{} \triangleq 
            \frac{d\sqrt{2K}}{3}
            \left(\frac{4N}{3}\right)^{d/4}_{}$, 
            $M \leq 
            \left\lceil
            \big(
                {\gamma}\sqrt{3N}
                {\eta}
                K^{2/d}_{}
                \big)^{\frac{2d}{4+d}}_{}
            \right\rceil
            $%
            ,
        \end{enumerate}
        where 
        %
        ${\eta} 
        \triangleq 
        \frac{\epsilon}{d}\sqrt{\frac{8N}{3}} \left(1+\sqrt{\frac{N}{3}}\right)
        $
        and
        ${\gamma} 
        \triangleq 
        \Big(1+\sqrt{\frac{N}{3}}\Big)^{\nspp-1}_{}
        $%
        .
    \end{theorem}
    \begin{proof}
        Note that
        we can 
        rewrite 
        %
        %
        $\xi(m) = m^{d/2+2}_{} - {\eta} {\gamma} K^{1/d}_{} m - {\eta}(1-{\gamma}) K^{2/d}_{} $.
        Let 
        $m_l^{}=  \big({\eta} K^{2/d}_{}\big)^{\frac{2}{4+d}}_{}$; 
        observing 
        that 
        $\epsilon \geq \epsilon_0^{} 
        \iff 
        {\eta} \geq \sqrt{K}$%
        , we 
        deduce
        $
        \xi(m_l^{}) 
        = 
        {\eta} K^{2/d}_{} - {\eta} {\gamma} K^{1/d}_{} \big( {\eta} K^{2/d}_{} \big)^{\frac{2}{4+d}}_{} - {\eta}(1-{\gamma})K^{2/d}_{}
        = 
        {\eta} {\gamma} K^{1/d}_{} \left( K^{1/d}_{} - \big( {\eta} K^{2/d}_{} \big)^{\frac{2}{4+d}}_{} \right) 
        \leq 
        {\eta} {\gamma} K^{1/d}_{} \Big( K^{1/d}_{} - \big( K^{1/2 + 2/d}_{} \big)^{\frac{2}{4+d}}_{} \Big) 
        = 
        0
        $.
        %
        %
        %
        %
        %
        Also,
        taking 
        $m_u^{}=  \big({\gamma}\sqrt{3N}{\eta} K^{2/d}_{}\big)^{\frac{2}{4+d}}_{}$ 
        and noting 
        $\epsilon \leq \epsilon_1^{} 
        \iff 
        {\eta} \leq \frac{2\sqrt{K}}{3{\gamma}} \left(\frac{4N}{3}\right)^{d/4+1/2}_{}$,  
        we 
        have
        $
        \xi(m_u^{}) 
        = 
        {\gamma}\sqrt{3N} {\eta} K^{2/d}_{} - {\eta} {\gamma} K^{1/d}_{} \big( {\gamma}\sqrt{3N} {\eta} K^{2/d}_{} \big)^{\frac{2}{4+d}}_{} - {\eta}(1-{\gamma})K^{2/d}_{}
        = 
        {\eta} {\gamma} K^{1/d}_{} \left( 2\sqrt{\frac{N}{3}} K^{1/d}_{} - \big({\gamma}\sqrt{3N} {\eta} K^{2/d}_{} \big)^{\frac{2}{4+d}}_{} \right) 
        \geq 
        {\eta} {\gamma} K^{1/d}_{} \left( 2\sqrt{\frac{N}{3}} K^{1/d}_{} - \left( \left(\frac{4N}{3}\right)^{d/4+1}_{} K^{1/2 + 2/d}_{} \right)^{\frac{2}{4+d}}_{} \right) 
        = 
        0
        $.
        %
        %
        %
        %
        %
        %
        %
        %
        By invoking Proposition \ref{prop:xi_root}, 
        we get 
        $M \geq \big\lfloor m_l^{d} \big\rfloor$ and $M \leq  \big\lceil m_u^{d} \big\rceil $
        under respective conditions.
    \end{proof}


    Table \ref{tab:results_eps0_eps1} provides the ranges of $\epsilon$ 
    for which the aforementioned bounds hold. 
    %
    It
    is evident that 
    the optimal grid size is 
    \begin{equation}
    M = \Theta\Big( (N \epsilon)^{\frac{2d}{4+d}} \, K^{\frac{4}{4+d}} \Big)
    \end{equation}
    for acceptable values of $\epsilon$ \cite{nist800226}. 
    Unlike EUGkM, this grid size depends explicitly on $K$ and have different exponents for 
    $N$
    and
    $\epsilon$.
    Also, 
    note that
    the grid size is independent of $r$.

    \begin{table}[h]
        \centering
        \caption{
        Ranges of ${\epsilon}$ 
        for which the bounds on $M$ in Theorem \ref{thm:M_bound} hold.%
        }
        \label{tab:results_eps0_eps1}
        \setlength{\extrarowheight}{1.5pt}
        \setlength{\tabcolsep}{4.5pt}        
        \setlength{\aboverulesep}{0pt}
        \setlength{\belowrulesep}{0pt}
        \setlength{\cmidrulekern}{0pt}
    
        
        \begin{tabular}{C{0.75em} ? C{4.25em} ? *{2}{C{4.25em}| C{4.25em} ?}}
        
            %
    
            \multicolumn{2}{c}{}
            & 
            \multicolumn{2}{c}{
            $d=2 \,$}
            & 
            \multicolumn{2}{c}{
            $d=3 \,$}
            \\
            %
            
            \cmidrule[1pt]{2-6}
            
            \multicolumn{1}{c?}{}  & 
            \begin{tabular}{c} $\boldsymbol{N}$\\[-0.1ex] \end{tabular}  &  
            \begin{tabular}{c} $\mbf{\epsilon_0^{}}$\\[0.55ex] \end{tabular}  &
            \begin{tabular}{c} $\mbf{\epsilon_1^{}}$\\[0.55ex] \end{tabular}  &
            \begin{tabular}{c} $\mbf{\epsilon_0^{}}$\\[0.55ex] \end{tabular}  & 
            \begin{tabular}{c} $\mbf{\epsilon_1^{}}$\\[0.55ex] \end{tabular} 
            \\
    
            \cmidrule[1pt]{2-6}
            \specialrule{0pt}{1pt}{0pt}	
            \cmidrule[1pt]{2-6}
    
            \multirow{3}{*}{\spheading{\centering{
            $\,\,\ \ K=2$}}} 
            
            &
            100
            &
            0.026  &  15.396
            & 
            0.039  &  78.475
            \\
            
            \cmidrule{2-6}
            
            &
            200
            &
            0.014  &  21.773
            &
            0.021  &  131.979
            \\
            
            \cmidrule{2-6}
            
            &
            400
            &
            0.007  &  30.792
            &
            0.011  &  221.962
            \\
            
            \cmidrule[1pt]{2-6}
            \specialrule{0pt}{1pt}{0pt}	
            \cmidrule[1pt]{2-6}
    
            \multirow{3}{*}{\spheading{\centering{
            $\,\,\ \ K=4$}}} 
            
            &
            200 
            &
            0.019  &  30.792
            &
            0.029  &  186.647
            \\
            
            \cmidrule{2-6}
            
            &
            400
            &
            0.010  &  43.546
            &
            0.015  &  313.901
            \\
            
            \cmidrule{2-6}
            
            &
            800 
            &
            0.005  &  61.584
            &
            0.008  &  527.918
            \\
    
            \cmidrule[1pt]{2-6}
            \specialrule{0pt}{1pt}{0pt}	
            \cmidrule[1pt]{2-6}
    
            \multirow{3}{*}{\spheading{\centering{
            $\,\,\ \ K=8$}}} 
            
            &
            400 
            &
            0.014  &  61.584
            &
            0.021  &  443.924
            \\
            
            \cmidrule{2-6}
            
            &
            800
            & 
            0.008  &  87.092
            &
            0.011  &  746.588
            \\
            
            \cmidrule{2-6}
            
            &
            1600
            & 
            0.004  &  123.168
            &
            0.006  &  1255.607
            \\
            
            \cmidrule[1pt]{2-6}
        \end{tabular}
    \end{table} 

    

    \section{Simulation Results}\label{sec:simulation}

    We evaluate the performance of our proposed scheme, RUGNIK, against several state-of-the-art DP $K$-means clustering methods, namely, EUGkM \cite{su2017differentially}, DPLloyd implementation in IBM’s \texttt{diffprivlib} \cite{holohan2019diffprivlib}, and Google’s coreset-based clustering (GCC) \cite{chang2021practical}. 
    All methods are tested on synthetic datasets generated as described in \cite{su2017differentially} for a variety of configurations.

    Table \ref{tab:results_wcss} summarizes the average WCSS achieved by different techniques across these settings for a wide range of privacy budgets $\epsilon$. Lower WCSS indicates better clustering quality. The results show that our 
    proposed 
    scheme consistently outperforms all baseline schemes across all considered settings. 
    Furthermore, RUGNIK offers a significant improvement over EUGkM that becomes more pronounced in regimes with stronger privacy constraints (i.e., smaller $\epsilon$) and larger numbers of clusters, where the choice of grid resolution plays a more significant role. 

    Table \ref{tab:results_n_grids} compares the number of grids $M$ selected by EUGkM and our proposed scheme. It is evident that EUGkM, due to its independence from the number of clusters $K$, results in degenerate choices in high-privacy regimes, where the number of grids is smaller than the number of clusters ($M < K$, marked with \MleqK), 
    thereby 
    rendering the clustering process illogical as the grid size is too coarse to distinguish the required number of 
    clusters%
    . 
    In contrast, RUGNIK appropriately scales the grid resolution by explicitly accounting for $K$, ensuring 
    high clustering accuracy 
    even under strong privacy constraints.

    \begin{table*}
        \centering
        \caption{
        Within-cluster sum of squares (WCSS) achieved by various methods.%
        }
        \label{tab:results_wcss}
        \setlength{\extrarowheight}{1.5pt}
        \setlength{\tabcolsep}{1pt}        
        \setlength{\aboverulesep}{0pt}
        \setlength{\belowrulesep}{0pt}
        \setlength{\cmidrulekern}{0pt}
    
        
        \begin{tabular}{ C{1.18em}  ? C{5.1em} | C{4.845em}
        +
        *{5}{C{4em}|} C{4em}
        +
        *{5}{C{4em}|} C{4em} ?}
    
            %
    
            \multicolumn{3}{c}{}
            & 
            \multicolumn{6}{c}{
            $d\!=\!2 \:\,$}
            &
            \multicolumn{6}{c}{
            $d\!=\!3 \,$}
            \\
            %
            
            \cmidrule[1pt]{2-15}
            
            & \multicolumn{2}{c+}{\textbf{Privacy budget, $\mbf{\epsilon}$}} & 
            \textbf{0.1} & \textbf{0.15} & \textbf{0.25} & \textbf{0.4} & \textbf{0.6} & \textbf{1.0} 
            &
            \textbf{0.1} & \textbf{0.15} & \textbf{0.25} & \textbf{0.4} & \textbf{0.6} & \textbf{1.0}
            \\
    
            \cmidrule[1pt]{2-15}
            \specialrule{0pt}{1pt}{0pt}	
            \cmidrule[1pt]{2-15}
    
            \multirow{15}{*}{\spheading{\centering{
            $K\!=\!2 \ \ \ $}}} 
            &
            
            \multirow{5}{*}{$N\!=\!100$} 
            
            & \textbf{IBM} & 
            28.556  &  28.014  &  26.775  &  24.757  &  22.498  &  18.500 
            & 
            45.597  &  44.860  &  43.051  &  40.357  &  36.978  &  30.918
            \\
            
            \cmidrule{3-15} 
            
            && \textbf{GCC} & 
            15.035  &  14.772  &  14.233  &  13.434  &  12.486  &  10.950 
            & 
            19.990  &  19.673  &  19.044  &  18.092  &  16.897  &  14.952  
            \\
            
            \cmidrule{3-15} 
            
            && \textbf{EUGkM} & 
            16.523  &  16.523  &  3.968  &  3.621  &  3.509  &  2.479
            & 
            21.355  &  21.355  &  21.355  &  6.931  &  6.413  &  \textbf{3.797}
            \\
            
            \cmidrule{3-15}
            
            && \textbf{RUGNIK} & 
            \textbf{6.114}  &  \textbf{6.002}  &  \textbf{3.824}  &  \textbf{2.417}  &  \textbf{1.862}  &  \textbf{1.537}
            & 
            \textbf{13.091}  &  \textbf{10.604}  &  \textbf{8.214}  &  \textbf{6.286}  &  \textbf{4.651}  &  3.800
            \\
            
            \cmidrule{3-15}
            
            && \textbf{No priv.} & 
            \multicolumn{6}{c+}{1.249}
            & 
            \multicolumn{6}{c?}{1.900}
            \\
            
            \cmidrule[1pt]{2-15}
            
            & \multirow{5}{*}{$N\!=\!200$} 
            
            & \textbf{IBM} & 
            54.555  &  52.172  &  47.163  &  41.049  &  33.716  &  23.594
            &
            87.687  &  84.612  &  77.367  &  67.751  &  57.027  &  42.012
            \\
            
            \cmidrule{3-15} 
            
            && \textbf{GCC} & 
            29.165  &  28.144  &  26.179  &  23.632  &  20.849  &  16.854
            &
            38.866  &  37.671  &  35.308  &  32.112  &  28.565  &  23.281
            \\
            
            \cmidrule{3-15} 
            
            && \textbf{EUGkM} & 
            33.190  &  7.483  &  6.977  &  5.010  &  4.882  &  2.927
            &
            42.708  &  14.980  &  12.883  &  12.182  &  7.294  &  6.885
            \\
            
            \cmidrule{3-15}
            
            && \textbf{RUGNIK} &
            \textbf{9.077}  &  \textbf{6.682}  &  \textbf{4.041}  &  \textbf{3.253}  &  \textbf{2.930}  &  \textbf{2.696}
            &
            \textbf{18.004}  &  \textbf{14.964}  &  \textbf{10.428}  &  \textbf{8.089}  &  \textbf{5.697}  &  \textbf{4.838}
            \\
            
            \cmidrule{3-15}
            
            && \textbf{No priv.} & 
            \multicolumn{6}{c+}{2.458}
            & 
            \multicolumn{6}{c?}{3.838}
            \\
            
            \cmidrule[1pt]{2-15}
            
            & \multirow{5}{*}{$N\!=\!400$} 
            
            & \textbf{IBM} & 
            99.167  &  90.287  &  74.488  &  56.616  &  40.557  &  23.960
            &
            161.077  &  148.092  &  124.883  &  98.466  &  74.151  &  46.804
            \\
            
            \cmidrule{3-15} 
            
            && \textbf{GCC} & 
            54.828  &  51.176  &  44.937  &  37.717  &  31.434  &  33.692
            &
            73.927  &  69.408  &  61.617  &  52.404  &  43.356  &  36.734
            \\
            
            \cmidrule{3-15} 
            
            && \textbf{EUGkM} & 
            14.302  &  13.848  &  9.766  &  6.054  &  5.310  &  \textbf{5.088}
            &
            27.464  &  25.472  &  \textbf{14.949}  &  13.869  &  9.390  &  8.966
            \\
            
            \cmidrule{3-15}
            
            && \textbf{RUGNIK} &
            \textbf{9.576}  &  \textbf{7.434}  &  \textbf{6.030}  &  \textbf{5.504}  &  \textbf{5.283}  &  5.112
            &
            \textbf{24.739}  &  \textbf{18.323}  &  14.962  &  \textbf{10.228}  &  \textbf{9.389}  &  \textbf{8.447}
            \\
            
            \cmidrule{3-15}
            
            && \textbf{No priv.} & 
            \multicolumn{6}{c+}{4.955}
            & 
            \multicolumn{6}{c?}{7.690}
            \\
            
            \cmidrule[1pt]{2-15}
            \specialrule{0pt}{1pt}{0pt}	
            \cmidrule[1pt]{2-15}
    
            \multirow{15}{*}{\spheading{\centering{
            $K\!=\!4 \ \ \ $}}} 
            & 
            
            \multirow{5}{*}{$N\!=\!200$} 
            
            & \textbf{IBM} &
            16.669  &  16.408  &  15.917  &  15.092  &  14.043  &  12.053
            &
            30.666  &  30.231  &  29.297  &  27.748  &  25.788  &  22.417
            \\
            
            \cmidrule{3-15} 
            
            && \textbf{GCC} &
            18.926  &  18.436  &  17.558  &  16.442  &  15.302  &  13.607
            &
            22.871  &  22.414  &  21.550  &  20.348  &  19.042  &  17.023
            \\
            
            \cmidrule{3-15} 
            
            && \textbf{EUGkM} &
            21.808  &  4.533  &  4.533  &  2.660  &  2.594  &  1.423
            &
            25.145  &  \textbf{8.140}  &  7.620  &  7.371  &  3.981  &  3.685
            \\
            
            \cmidrule{3-15}
            
            && \textbf{RUGNIK} & 
            \textbf{4.508}  &  \textbf{3.249}  &  \textbf{2.000}  &  \textbf{1.357}  &  \textbf{1.032}  &  \textbf{0.844}
            &
            \textbf{8.849}  &  8.309  &  \textbf{5.920}  &  \textbf{4.493}  &  \textbf{3.030}  &  \textbf{2.184}
            \\
            
            \cmidrule{3-15}
            
            && \textbf{No priv.} & 
            \multicolumn{6}{c+}{0.683}
            & 
            \multicolumn{6}{c?}{1.319}
            \\
            
            \cmidrule[1pt]{2-15}
            
            & \multirow{5}{*}{$N\!=\!400$} 
            
            & \textbf{IBM} &
            32.278  &  31.015  &  29.130  &  26.104  &  22.679  &  17.309
            &
            58.917  &  57.131  &  53.059  &  47.888  &  42.125  &  33.462
            \\
            
            \cmidrule{3-15} 
            
            && \textbf{GCC} &
            36.170  &  34.569  &  31.911  &  28.971  &  26.022  &  21.663
            &
            43.900  &  42.285  &  39.422  &  35.987  &  32.513  &  27.311
            \\
            
            \cmidrule{3-15} 
            
            && \textbf{EUGkM} &
            9.133  &  9.132  &  5.218  &  3.012  &  1.890  &  1.624
            &
            15.708  &  15.125  &  8.276  &  7.443  &  4.420  &  4.082
            \\
            
            \cmidrule{3-15}
        
            && \textbf{RUGNIK} &
            \textbf{5.259}  &  \textbf{3.319}  &  \textbf{2.297}  &  \textbf{1.797}  &  \textbf{1.654}  &  \textbf{1.515}
            &
            \textbf{13.693}  &  \textbf{10.488}  &  \textbf{7.124}  &  \textbf{5.099}  &  \textbf{3.884}  &  \textbf{3.345}
            \\
            
            \cmidrule{3-15}
            
            && \textbf{No priv.} & 
            \multicolumn{6}{c+}{1.382}
            & 
            \multicolumn{6}{c?}{2.662}
            \\
            
            \cmidrule[1pt]{2-15}
            
            & \multirow{5}{*}{$N\!=\!800$} 
            
            & \textbf{IBM} & 
            60.079  &  55.842  &  48.604  &  39.834  &  30.918  &  20.287
            &
            109.706  &  102.227  &  89.782  &  75.045  &  60.635  &  43.174
            \\
            
            \cmidrule{3-15} 
            
            && \textbf{GCC} & 
            66.463  &  61.989  &  55.171  &  47.790  &  40.181  &  35.657
            &
            82.116  &  76.997  &  68.908  &  59.966  &  50.976  &  39.202
            \\
            
            \cmidrule{3-15} 
            
            && \textbf{EUGkM} & 
            10.938  &  10.649  &  6.087  &  3.350  &  3.068  &  \textbf{2.930}
            &
            28.825  &  15.879  &  14.798  &  8.231  &  6.405  &  \textbf{5.853}
            \\
            
            \cmidrule{3-15}
            
            && \textbf{RUGNIK} &
            \textbf{5.481}  &  \textbf{4.153}  &  \textbf{3.427}  &  \textbf{3.166}  &  \textbf{3.051}  &  2.947
            &
            \textbf{18.022}  &  \textbf{12.075}  &  \textbf{9.173}  &  \textbf{6.985}  &  \textbf{6.336}  &  5.865
            \\
            
            \cmidrule{3-15}
            
            && \textbf{No priv.} & 
            \multicolumn{6}{c+}{2.815}
            & 
            \multicolumn{6}{c?}{5.367}
            \\
    
            \cmidrule[1pt]{2-15}
            \specialrule{0pt}{1pt}{0pt}	
            \cmidrule[1pt]{2-15}
    
            \multirow{15}{*}{\spheading{\centering{
            $K\!=\!8 \ \ \ $}}} 
            & 
    
            \multirow{5}{*}{$N\!=\!400$} 
            
            & \textbf{IBM} & 
            8.981  &  8.903  &  8.664  &  8.343  &  7.898  &  7.063
            &
            20.068  &  19.835  &  19.366  &  18.585  &  17.603  &  15.829
            \\
            
            \cmidrule{3-15} 
            
            && \textbf{GCC} & 
            20.018  &  19.134  &  17.668  &  16.065  &  14.540  &  12.479
            &
            26.408  &  25.543  &  24.034  &  22.298  &  20.531  &  18.073
            \\
            
            \cmidrule{3-15} 
            
            && \textbf{EUGkM} & 
            6.408  &  6.408  &  2.431  &  2.050  &  1.331  &  0.874
            &
            8.361  &  8.352  &  5.347  &  5.044  &  3.187  &  2.986
            \\
            
            \cmidrule{3-15}
            
            && \textbf{RUGNIK} & 
            \textbf{2.723}  &  \textbf{1.944}  &  \textbf{1.268}  &  \textbf{0.785}  &  \textbf{0.601}  &  \textbf{0.481}
            &
            \textbf{6.975}  &  \textbf{6.383}  &  \textbf{4.564}  &  \textbf{3.341}  &  \textbf{2.395}  &  \textbf{1.721}
            \\
            
            \cmidrule{3-15}
            
            && \textbf{No priv.} & 
            \multicolumn{6}{c+}{0.362}
            & 
            \multicolumn{6}{c?}{0.946}
            \\
            
            \cmidrule[1pt]{2-15}
            
            & \multirow{5}{*}{$N\!=\!800$} 
            
            & \textbf{IBM} & 
            17.419  &  16.989  &  16.176  &  14.889  &  13.410  &  10.952
            &
            38.841  &  37.896  &  36.046  &  33.368  &  30.134  &  24.807
            \\
            
            \cmidrule{3-15} 
            
            && \textbf{GCC} & 
            37.018  &  34.490  &  30.835  &  27.167  &  23.624  &  18.502
            &
            49.828  &  47.180  &  43.113  &  38.757  &  34.549  &  28.448
            \\
            
            \cmidrule{3-15} 
            
            && \textbf{EUGkM} & 
            4.800  &  4.783  &  4.253  &  1.761  &  1.280  &  0.835
            &
            16.616  &  10.258  &  9.796  &  5.923  &  3.607  &  2.630
            \\
            
            \cmidrule{3-15}
            
            && \textbf{RUGNIK} & 
            \textbf{3.082}  &  \textbf{2.012}  &  \textbf{1.322}  &  \textbf{0.972}  &  \textbf{0.865}  &  \textbf{0.794}
            &
            \textbf{10.532}  &  \textbf{8.159}  &  \textbf{5.461}  &  \textbf{3.888}  &  \textbf{3.074}  &  \textbf{2.524}
            \\
            
            \cmidrule{3-15}
            
            && \textbf{No priv.} & 
            \multicolumn{6}{c+}{0.701}
            & 
            \multicolumn{6}{c?}{1.847}
            \\
            
            \cmidrule[1pt]{2-15}
            
            & \multirow{5}{*}{$N\!=\!1600$} 
            
            & \textbf{IBM} & 
            33.091  &  31.331  &  28.304  &  24.307  &  20.137  &  14.435
            &
            73.421  &  69.868  &  63.186  &  54.705  &  45.614  &  33.253
            \\
            
            \cmidrule{3-15} 
            
            && \textbf{GCC} & 
            65.587  &  59.517  &  51.077  &  42.144  &  33.696  &  27.025
            &
            90.567  &  83.593  &  73.680  &  63.118  &  52.809  &  39.647
            \\
            
            \cmidrule{3-15} 
            
            && \textbf{EUGkM} & 
            8.351  &  5.227  &  3.454  &  2.144  &  1.595  &  \textbf{1.497}
            &
            20.103  &  12.695  &  11.826  &  6.868  &  5.071  &  \textbf{4.097}
            \\
            
            \cmidrule{3-15}
            
            && \textbf{RUGNIK} &
            \textbf{3.060}  &  \textbf{2.286}  &  \textbf{1.869}  &  \textbf{1.662}  &  \textbf{1.587}  &  1.539
            &
            \textbf{13.341}  &  \textbf{9.611}  &  \textbf{6.609}  &  \textbf{5.085}  &  \textbf{4.564}  &  4.136
            \\
            
            \cmidrule{3-15}
            
            && \textbf{No priv.} & 
            \multicolumn{6}{c+}{1.409}
            & 
            \multicolumn{6}{c?}{3.718}
            \\
            
            \cmidrule[1pt]{2-15}
        \end{tabular}
    \end{table*}

    \begin{table*}
        \centering
        \caption{
        Number of grids selected by {\normalfont EUGkM} and RUGNIK.%
        }
        \label{tab:results_n_grids}
        \setlength{\extrarowheight}{1.5pt}
        \setlength{\tabcolsep}{4.5pt}        
        \setlength{\aboverulesep}{0pt}
        \setlength{\belowrulesep}{0pt}
        \setlength{\cmidrulekern}{0pt}
    
        
        \begin{tabular}{C{0.75em}  ? c | C{4em}
        +
        *{5}{C{2em}|} C{2em}
        +
        *{5}{C{2em}|} C{2em} ?}
    
            %
    
            \multicolumn{3}{c}{}
            & 
            \multicolumn{6}{c}{
            $d=2 \:\,$}
            & 
            \multicolumn{6}{c}{
            $d=3 \,$}
            \\
            %
            
            \cmidrule[1pt]{2-15}
            
            & \multicolumn{2}{c+}{\textbf{Privacy budget, $\mbf{\epsilon}$}} & 
            \textbf{0.1} & \textbf{0.15} & \textbf{0.25} & \textbf{0.4} & \textbf{0.6} & \textbf{1.0} 
            &
            \textbf{0.1} & \textbf{0.15} & \textbf{0.25} & \textbf{0.4} & \textbf{0.6} & \textbf{1.0}
            \\
    
            \cmidrule[1pt]{2-15}
            \specialrule{0pt}{1pt}{0pt}	
            \cmidrule[1pt]{2-15}
    
            \multirow{6}{*}{\spheading{\centering{
            $K=2 \;\ $}}} 
            &
            
            \multirow{2}{*}{$N\!=\!100$} 
            
            & \textbf{EUGkM} & 
            1\MleqK  &  1\MleqK  &  4  &  4  &  4  &  9 
            & 
            1\MleqK  &  1\MleqK  &  1\MleqK  &  8  &  8  &  27
            \\
            
            \cmidrule{3-15}
            
            && \textbf{RUGNIK} & 
            4  &  9  &  9  &  16  &  16  &  25
            & 
            8  &  8  &  8  &  27  &  27  &  27
            \\
            
            \cmidrule[1pt]{2-15}
            
            & \multirow{2}{*}{$N\!=\!200$} 
            
            & \textbf{EUGkM} & 
            1\MleqK  &  4  &  4  &  9  &  9  &  16
            &
            1\MleqK  &  8  &  8  &  8  &  27  &  27
            \\
            
            \cmidrule{3-15}
            
            && \textbf{RUGNIK} &
            9  &  9  &  16  &  25  &  36  &  49
            &
            8  &  8  &  27  &  27  &  64  &  64
            \\
            
            \cmidrule[1pt]{2-15}
            
            & \multirow{2}{*}{$N\!=\!400$} 
            
            & \textbf{EUGkM} & 
            4  &  4  &  9  &  16  &  25  &  36
            &
            8  &  8  &  27  &  27  &  64  &  64
            \\
            
            \cmidrule{3-15}
            
            && \textbf{RUGNIK} &
            16  &  16  &  25  &  36  &  49  &  64
            &
            27  &  27  &  27  &  64  &  64  &  125
            \\
            
            \cmidrule[1pt]{2-15}
            \specialrule{0pt}{1pt}{0pt}	
            \cmidrule[1pt]{2-15}
    
            \multirow{6}{*}{\spheading{\centering{
            $K=4  \;\ $}}} 
            & 
            
            \multirow{2}{*}{$N\!=\!200$} 
            
            & \textbf{EUGkM} &
            1\MleqK  &  4  &  4  &  9  &  9  &  16
            &
            1\MleqK  &  8  &  8  &  8  &  27  &  27
            \\
            
            \cmidrule{3-15}
            
            && \textbf{RUGNIK} & 
            16  &  16  &  25  &  36  &  49  &  64
            &
            8  &  27  &  27  &  64  &  64  &  125
            \\
            
            \cmidrule[1pt]{2-15}
            
            & \multirow{2}{*}{$N\!=\!400$} 
            
            & \textbf{EUGkM} &
            4  &  4  &  9  &  16  &  25  &  36
            &
            8  &  8  &  27  &  27  &  64  &  64
            \\
            
            \cmidrule{3-15}
        
            && \textbf{RUGNIK} &
            25  &  25  &  36  &  49  &  81  &  100
            &
            27  &  27  &  64  &  64  &  125  &  216
            \\
            
            \cmidrule[1pt]{2-15}
            
            & \multirow{2}{*}{$N\!=\!800$} 
            
            & \textbf{EUGkM} & 
            9  &  9  &  16  &  36  &  49  &  81
            &
            8  &  27  &  27  &  64  &  125  &  216
            \\
            
            \cmidrule{3-15}
            
            && \textbf{RUGNIK} &
            36  &  36  &  64  &  81  &  121  &  169
            &
            27  &  64  &  64  &  125  &  216  &  343
            \\
    
            \cmidrule[1pt]{2-15}
            \specialrule{0pt}{1pt}{0pt}	
            \cmidrule[1pt]{2-15}
    
            \multirow{6}{*}{\spheading{\centering{
            $K=8  \;\ $}}} 
            & 
    
            \multirow{2}{*}{$N\!=\!400$} 
            
            & \textbf{EUGkM} & 
            4\MleqK  &  4\MleqK  &  9  &  16  &  25  &  36
            &
            8  &  8  &  27  &  27  &  64  &  64
            \\
            
            \cmidrule{3-15}
            
            && \textbf{RUGNIK} & 
            36  &  36  &  64  &  81  &  121  &  169
            &
            27  &  64  &  64  &  125  &  125  &  216
            \\
            
            \cmidrule[1pt]{2-15}
            
            & \multirow{2}{*}{$N\!=\!800$} 
            
            & \textbf{EUGkM} & 
            9  &  9  &  16  &  36  &  49  &  81
            &
            8  &  27  &  27  &  64  &  125  &  216
            \\
            
            \cmidrule{3-15}
            
            && \textbf{RUGNIK} & 
            49  &  64  &  100  &  121  &  169  &  256
            &
            64  &  64  &  125  &  216  &  343  &  512
            \\
            
            \cmidrule[1pt]{2-15}
            
            & \multirow{2}{*}{$N\!=\!1600$} 
            
            & \textbf{EUGkM} & 
            16  &  25  &  36  &  64  &  100  &  169
            &
            27  &  64  &  64  &  125  &  216  &  512
            \\
            
            \cmidrule{3-15}
            
            && \textbf{RUGNIK} &
            81  &  100  &  144  &  196  &  289  &  400
            &
            125  &  125  &  216  &  343  &  512  &  729
            \\
            
            \cmidrule[1pt]{2-15}
        \end{tabular}
    \end{table*} 

    
    
    \section{Conclusions}\label{sec:conclusion}

    In this work, we 
    studied the problem of non-interactive differentially private $K$-means clustering, 
    and 
    we derived the optimal 
    %
    grid resolution for uniform gridding%
    .
    Our selection rule is theoretically grounded and explicitly incorporates the number of clusters
    in the grid size%
    . The efficacy and robustness of our refined grid resolution are validated using extensive 
    empirical analysis
    under a wide range of settings%
    , and we demonstrated superior performance over the existing state-of-the-art and benchmark open-source implementations 
    such as \texttt{diffprivlib} \cite{holohan2019diffprivlib}, and Google’s coreset-based clustering \cite{chang2021practical}%
    , with significant improvements in high privacy regime. 
    Future work will focus on extending the analysis to adaptive grid structures and on 
    discrete noise mechanisms
    that will further improve performance in the high-privacy regime.

    {
    	\bibliographystyle{IEEEtran}
    	\bibliography{refs}
    }
    
    \end{document}